\documentclass[11pt]{article}

\usepackage{amssymb}
\usepackage{amsfonts}
\usepackage{amsmath}
\usepackage{amsthm}
\usepackage[english]{babel}
\usepackage{latexsym}
\usepackage{color}
\usepackage{enumerate}
\usepackage{nicefrac}
\usepackage{mathrsfs}
\usepackage{comment}
\usepackage[hmargin=2cm,vmargin=2cm]{geometry}
\usepackage{bm}
\usepackage{centernot}

\usepackage[affil-it]{authblk}
\theoremstyle{plain}
\newtheorem{theorem}{Theorem}[section]
\newtheorem{lemma}[theorem]{Lemma}
\newtheorem{proposition}[theorem]{Proposition}

\theoremstyle{definition}
\newtheorem{definition}[theorem]{Definition}

\newtheorem{remark}[theorem]{Remark}
\theoremstyle{remark}

\numberwithin{equation}{section}

\newcommand{\E}{\mathbb{E}}

\newcommand{\probp}{\mathbb{P}}

\newcommand{\R}{\mathbb{R}}
\newcommand{\N}{\mathbb{N}}

\newcommand{\cF}{{\mathcal{F}}}

\newcommand{\cC}{\mathcal{C}}

\newcommand{\VaR}{\mathop {\rm VaR}\nolimits}

\newcommand{\ES}{\mathop {\rm ES}\nolimits}

\newcommand{\abs}[1]{\lvert#1\rvert}
\newcommand{\norm}[1]{\lVert#1\rVert}

\def\keywords{\vspace{.5em}
{\noindent\textbf{Keywords}:\,\relax%
}}

\def\JELclassification{\vspace{.5em}
{\noindent\textbf{JEL classification}:\,\relax%
}}

\def\MSCclassification{\vspace{.5em}
{\noindent\textbf{MSC}:\,\relax%
}}

\makeatletter
\def\@fnsymbol#1{\ensuremath{\ifcase#1\or 1\or 2\or 3\or 4\or 5\or 6\or 7\or
8\else\@ctrerr\fi}}
\makeatother

\begin{document}

\title{Stability properties of Haezendonck-Goovaerts premium principles}

\author{
\sc{Niushan Gao}\,,
\sc{Foivos Xanthos}
}
\affil{Department of Mathematics, Ryerson University, Toronto, Canada\\
{\tt niushan@ryerson.ca}\,, {\tt foivos@ryerson.ca}}
\author{
\sc{Cosimo Munari}
}
\affil{Center for Finance and Insurance and Swiss Finance Institute, University of Zurich,
Switzerland\\
{\tt cosimo.munari@bf.uzh.ch}}

\date{\today}

\maketitle

\parindent 0em \noindent

\begin{abstract}
\noindent We investigate a variety of stability properties of Haezendonck-Goovaerts premium principles on their natural domain, namely Orlicz spaces. We show that such principles always satisfy the Fatou property. This allows to establish a tractable dual representation without imposing any condition on the reference Orlicz function. In addition, we show that Haezendonck-Goovaerts principles satisfy the stronger Lebesgue property if and only if the reference Orlicz function fulfills the so-called $\Delta_2$ condition. We also discuss (semi)continuity properties with respect to $\Phi$-weak convergence of probability measures. In particular, we show that Haezendonck-Goovaerts principles, restricted to the corresponding Young class, are always lower semicontinuous with respect to the $\Phi$-weak convergence.
\end{abstract}


\JELclassification{C65, G10}

\MSCclassification{91B30, 91G80, 46E30}


\section{Introduction}

Premium principles based on Orlicz norms were introduced in the classical paper by Haezendonck and Goovaerts~\cite{HaezendonckGoovaerts1982} as multiplicative equivalents of the zero utility principle. The original formulation was extended by Goovaerts et al.~\cite{GoovaertsKaasDhaeneTang2004} and by Bellini and Rosazza Gianin~\cite{BelliniRosazza2008} to account for general (not necessarily positive) losses. The extended formulation is known under the name of {\em Haezendonck-Goovaerts premium principle} (or equivalently {\em Haezendonck-Goovaerts risk measure}) and has been the subject of extensive study in the literature, see Bellini and Rosazza Gianin~\cite{BelliniRosazza2009} and \cite{BelliniRosazza2012}, Goovaerts et al.~\cite{GoovaertsKaasDhaeneTang2012}, Mao and Hu~\cite{MaoHu2012}, Tang and Yang~\cite{TangYang2012} and \cite{TangYang2014}, Zhu et al.~\cite{ZhuZhangZhang2013}, Ahn and Shyamalkumar~\cite{AhnShyamalkumar2014}, Peng et al.~\cite{PengWangZheng2015}, Wang and Peng~\cite{WangPeng2016}, Liu et al.~\cite{LiuPengWang2017}, Wang et al.~\cite{WangLiuHouPeng2018}. We also refer to the general survey by Goovaerts and Laeven~\cite{GoovaertsLaeven2008} and to the recent paper by Bellini et al.~\cite{BelliniLaevenRosazza2018} where an axiomatization of Orlicz premia in terms of acceptance sets is established.

\smallskip

As already pointed out in the original work by Haezendonck and Goovaerts~\cite{HaezendonckGoovaerts1982}, the natural model space for such premium principles are Orlicz spaces. Perhaps influenced by the earlier framework of risk measure theory, the bulk of the later literature following~\cite{HaezendonckGoovaerts1982}  has developed the theory of Haezendonck-Goovaerts premium principles for bounded losses. In Bellini and Rosazza Gianin~\cite{BelliniRosazza2012} the authors go beyond the bounded setting and focus on a general Orlicz space. However, their main results, most notably their dual representation, are only valid in a special class of Orlicz spaces, which are characterized by Orlicz functions satisfying the growth condition known as {\em $\Delta_2$ condition} in the literature. In this special case, the Orlicz space allows for a tractable duality theory, see e.g.\ Cheridito and Li~\cite{CheriditoLi2008,CheriditoLi2009}.

\smallskip

Since the Orlicz function defining the premium principle in its original formulation originates from a general (von Neumann-Morgenstern) utility function, there is {\em a priori} no legitimate restriction on the Orlicz function and, thus, on the corresponding Orlicz space. In this short note, we aim to fill this gap and provide a full picture on the stability properties and dual representation of Haezendonck-Goovaerts premium principles defined on a general Orlicz space. Most notably, we prove that Haezendonck-Goovaerts premium principles always satisfy the Fatou property and characterize when they satisfy the stronger Lebesgue property. As a consequence, we can exploit some recent results on functionals on Orlicz spaces, see Gao et al.~\cite{GaoLeungMunariXanthos2017}, to establish tractable dual representations of Haezendonck-Goovaerts premium principles without requiring the $\Delta_2$ condition. In addition, we investigate (semi)continuity properties with respect to the $\Phi$-weak convergence of probability measures. We show that $\Phi$-weak lower semicontinuity is always fulfilled on the appropriate Young domain and establish that $\Phi$-weak continuity holds if and only if $\Phi$ satisfies the $\Delta_2$ condition.


\section{Haezendonck-Goovaerts premium principles on Orlicz spaces}

We first provide a brief review of Orlicz spaces. For more details on Orlicz spaces we refer to Edgar and Sucheston~\cite{EdgarSucheston1992} and Zaanen~\cite{Zaanen1983}. Throughout the note, let $(\Omega,\cF,\probp)$ be a nonatomic probability space. We denote by $L^0$ the set of Borel measurable functions $X:\Omega\to\R$. As usual, we identify two functions that are equal almost surely. The set $L^0$ is equipped with its canonical vector lattice structure. A nonconstant function $\Phi:[0,\infty)\to[0,\infty]$ is said to be an {\em Orlicz function} if it is convex, nondecreasing, left-continuous, and satisfies $\Phi(0)=0$. The {\em conjugate} of $\Phi$ is the map $\Psi:[0,\infty)\to[0,\infty]$ defined by
\[
\Psi(y) := \sup_{x\in[0,\infty)}\{xy-\Phi(x)\}.
\]
The function $\Psi$ is also an Orlicz function. The {\em Orlicz space} associated with $\Phi$ is defined by
\[
L^\Phi := \left\{X\in L^0 \,; \ \E\left[\Phi\left(\frac{|X|}{\lambda}\right)\right]<\infty \ \ \mbox{for some} \ \lambda\in(0,\infty)\right\}.
\]
The {\em Orlicz heart} of $L^\Phi$ is given by
\[
H^\Phi := \left\{X\in L^\Phi \,; \ \E\left[\Phi\left(\frac{|X|}{\lambda}\right)\right]<\infty \ \ \mbox{for every} \ \lambda\in(0,\infty)\right\}.
\]
The Orlicz space $L^\Phi$ is a Banach lattice with respect to the Luxemburg norm
\[
\|X\|_\Phi := \inf\left\{\lambda\in(0,\infty) \,; \ \E\left[\Phi\left(\frac{|X|}{\lambda}\right)\right]\leq1\right\}.
\]
We always equip the Orlicz space $L^\Psi$ with the Orlicz norm
\[
\|Y\|_\Psi := \sup\{|\E[XY]| \,; \ X\in L^\Phi, \ \|X\|_\Phi\leq1\}.
\]
The subspace of $L^0$ consisting of $\probp$-bounded functions, respectively $\probp$-integrable functions, is denoted by $L^\infty$, respectively $L^1$. We always have $L^\infty\subset L^\Phi\subset L^1$ with norm-continuous embeddings. The norm dual of $L^\Phi$ is therefore a subspace of the norm dual of $L^\infty$ and in general cannot be identified with a function space. On the contrary, the norm dual of $H^\Phi$ can always be identified with $L^\Psi$ provided that $\Phi$ is finitely valued (see Edgar and Sucheston~\cite[Theorem 2.2.11]{EdgarSucheston1992}).

\smallskip

We say that $\Phi$ satisfies the {\em $\Delta_2$ condition} if there exist $y\in(0,\infty)$ and $k\in(0,\infty)$ such that $\Phi(2x)<k\Phi(x)$ for every $x\in[y,\infty)$. Since the underlying probability space is nonatomic, we have that $\Phi$ satisfies the $\Delta_2$ condition if and only if $L^\Phi=H^\Phi$ (see Edgar and Sucheston~\cite[Theorem 2.1.17]{EdgarSucheston1992}).

\smallskip

From now on, as is standard in the literature on Haezendonck-Goovaerts principles, we fix a finite-valued Orlicz function $\Phi$ that is normalized by $\Phi(1)=1$ (in the case that $\Phi$  assumes the value $\infty$ we have $L^\Phi=L^\infty$, for which we refer to the thorough study by Bellini and Rosazza Gianin~\cite{BelliniRosazza2008}). For a given $\alpha\in(0,1)$ we consider the Orlicz function $\Phi_\alpha:=\frac{\Phi}{1-\alpha}$ and define the map $N_\alpha:L^\Phi\to[0,\infty)$ by
\[
N_\alpha(X) := \|X\|_{\Phi_\alpha}.
\]
For later convenience, we introduce for every $X\in L^\Phi$ the function $\pi_\alpha(X,\cdot):\R\to\R$ defined by
\[
\pi_\alpha(X,m) := m+N_\alpha((X-m)^+)
\]
where $X^+:=\max(X,0)$ for every $X\in L^\Phi$.

\begin{definition}
The {\em Haezendonck-Goovaerts premium principle} associated to $\Phi$ at level $\alpha$ is the map $\pi_\alpha:L^\Phi\to\R$ defined by
\[
\pi_\alpha(X) := \inf_{m\in\R}\{m+N_\alpha((X-m)^+)\} = \inf_{m\in\R}\pi_\alpha(X,m).
\]
\end{definition}

\smallskip

The following result collects a number of useful properties of Haezendonck-Goovaerts premium principles.

\begin{proposition}[{\cite[Proposition 3]{BelliniRosazza2012}}]
\label{prop: preliminary result}
For every $X\in L^\Phi$ the function $\pi_\alpha(X,\cdot)$ is convex and satisfies
\[
\lim_{m\to-\infty}\pi_\alpha(X,m) = \lim_{m\to\infty}\pi_\alpha(X,m) = \infty.
\]
In particular, $\pi_\alpha$ is finitely valued and for every $X\in L^\Phi$ there exists $m_X\in\R$ such that
\[
\pi_\alpha(X) = \pi_\alpha(X,m_X) = m_X+N_\alpha((X-m_X)^+).
\]
Moreover, $\pi_\alpha$ is sublinear, monotone increasing, translation invariant, and law invariant.
\end{proposition}


\section{Stability properties}

In the spirit of Haezendonck and Goovaerts~\cite[Theorem 4]{HaezendonckGoovaerts1982}, in this section we take up the study of a variety of continuity properties of Haezendonck-Goovaerts premium principles.

\smallskip

We start by showing that Haezendonck-Goovaerts premium principles are always Lipschitz continuous with respect to the Orlicz norm induced by $\Phi$. This property was established by Bellini and Rosazza Gianin~\cite[Proposition 4]{BelliniRosazza2012} in the special setting of Orlicz hearts using the dual representation of Haezendonck-Goovaerts premium principles. Since $\pi_\alpha$ is sublinear, Lipschitz continuity on a general Orlicz space follows from the general results by Farkas et al.~\cite[Proposition 3.6, Theorem 3.16]{FarkasKochMunari2014}.

\begin{proposition}
The Haezendonck-Goovaerts premium principle $\pi_\alpha$ is Lipschitz continuous. In particular, for all $X,Y\in L^\Phi$ we have
\[
|\pi_\alpha(X)-\pi_\alpha(Y)| \leq \frac{1}{1-\alpha}\|X-Y\|_\Phi. \
\]
\end{proposition}
\begin{proof}
Since $1-\alpha\in(0,1)$, the convexity of $\Phi$ implies that $\Phi((1-\alpha)x)\leq(1-\alpha)\Phi(x)$ for every $x\in[0,\infty)$. As a result, for every $X\in L^\Phi$ and every $\lambda\in(0,\infty)$ we have
\[
\E\left[\Phi\left(\frac{|X|}{\lambda}\right)\right]\leq1 \ \implies \ \E\left[\Phi_\alpha\left(\frac{(1-\alpha)|X|}{\lambda}\right)\right]\leq1,
\]
showing that $N_\alpha(X)\leq\frac{1}{1-\alpha}\|X\|_\Phi$. The monotonicity of $\pi_\alpha$ now yields
\[
\pi_\alpha(X) \leq \pi_\alpha(|X|) \leq N_\alpha(|X|) \leq \frac{1}{1-\alpha}\|X\|_\Phi
\]
for every $X\in L^\Phi$. Then, it follows from subadditivity of $\pi_\alpha$ that
\[
\pi_\alpha(X)-\pi_\alpha(Y) \leq \pi_\alpha(X-Y) \leq \frac{1}{1-\alpha}\|X-Y\|_\Phi.
\]
for all $X,Y\in L^\Phi$. This establishes the desired statement.
\end{proof}

\smallskip

\begin{remark}
The above proof shows that $\pi_\alpha$ is Lipschitz continuous with constant $1$, i.e.\ is nonexpansive, with respect to the norm $N_\alpha$. This is in line with the general result by Pichler~\cite[Corollary 3.7]{Pichler2017}.
\end{remark}

\smallskip

As Lipschitz continuity is generally not strong enough to imply tractable dual representations, we turn to study continuity properties  with respect to the  {\em order structure} of $L^\Phi$. More precisely, we focus on the Fatou property, which corresponds to lower semicontinuity with respect to dominated almost-sure convergence, and on the stronger Lebesgue property, which corresponds to continuity with respect to dominated almost-sure convergence.

\smallskip

We start by showing that Haezendonck-Goovaerts principles are always continuous from below and satisfy the Fatou property. In fact, we show that they satisfy a stronger version of the Fatou property where the domination condition is not required at the level of the random variables but only of the corresponding expectations. We write $X_n\uparrow X$ to mean that the sequence $(X_n)\subset L^\Phi$ is increasing and converges almost surely to $X\in L^\Phi$.

\begin{theorem}
\label{prop: fatou}
The Haezendonck-Goovaerts premium principle $\pi_\alpha$ satisfies the following properties:
\begin{enumerate}[(i)]
  \item $\pi_\alpha$ is continuous from below, i.e.\ for every sequence $(X_n)\subset L^\Phi$ and every $X\in L^\Phi$ we have
\[
X_n\uparrow X \ \implies \ \pi_\alpha(X_n)\to\pi_\alpha(X).
\]
  \item $\pi_\alpha$ satisfies the Fatou property, i.e.\ for every sequence $(X_n)\subset L^\Phi$ such that $\sup|X_n|\in L^\Phi$ and for every $X\in L^\Phi$ we have
\[
X_n\xrightarrow{a.s.}X \ \implies \ \pi_\alpha(X)\leq\liminf\pi_\alpha(X_n).
\]
\item For every sequence $(X_n)\subset L^\Phi$ such that $\inf_{n\in\N}\mathbb{E}[X_n]>-\infty$ and for every $X\in L^\Phi$ we have
\[
X_n\xrightarrow{a.s.}X \ \implies \ \pi_\alpha(X)\leq\liminf\pi_\alpha(X_n).
\]
\end{enumerate}
\end{theorem}
\begin{proof}
It is well known that {\em (i)} and {\em (ii)} are equivalent for any monotone functional like $\pi_\alpha$. Note that there exists a constant $k\in\R$ such that $\mathbb{E}[\abs{X}]\leq k \norm{X}_{\Phi}$ for every $X \in L^{\Phi}$; see, e.g., Edgar and Sucheston~\cite[Proposition 2.2.1]{EdgarSucheston1992}. Thus, clearly, {\em (iii)} implies {\em (ii)}. It remains to establish {\em (iii)}. We observe first that, for all $X \in L^{\Phi}$ and $m \in \mathbb{R}$, Jensen's inequality yields that $(\mathbb{E}[X]-m)^+\leq \mathbb{E}[(X-m)^+]$. Hence, by Edgar and Sucheston~\cite[Corollary 2.3.11]{EdgarSucheston1992}, we have
\begin{equation}\label{DMeqaa}
N_{\alpha}((\mathbb{E}[X]-m)^+) \leq N_{\alpha}(\E[(X-m)^+])\leq N_\alpha((X-m)^+).
\end{equation}
This implies that for every $X\in L^\Phi$
\begin{equation}\label{DMeq}
\pi_{\alpha}(\mathbb{E}[X])\leq \pi_{\alpha}(X).
\end{equation}
Now, let $(X_n)\subset L^\Phi$ and $X \in L^\Phi$ be such that $\inf_n\mathbb{E}[X_n]>-\infty$ and $X_n \xrightarrow{{a.s.}} X$. If $\liminf_n\pi_\alpha(X_n)=\infty$, then there is nothing to prove. Thus, by passing to a subsequence, we may assume that
\[
\pi_\alpha(X_n) \to \liminf\pi_\alpha(X_n) \in [-\infty,\infty).
\]
We claim that $\sup_n\abs{\mathbb{E}[X_n]}<\infty$. If this is not the case, then we find a subsequence $(X_{n_k})$ such that $\mathbb{E}[X_{n_k}]\rightarrow \infty$. This is because $\inf_n\mathbb{E}[X_n]>-\infty$. By \eqref{DMeq} and the positive homogeneity of $\pi_{\alpha}$, it follows that $\pi_{\alpha}(X_{n_k})\geq \pi_\alpha(\E[X_{n_k}])= \mathbb{E}[X_{n_k}]\pi_{\alpha}(1)=\mathbb{E}[X_{n_k}]$ for $k$ large enough. As a result, we get $\pi_{\alpha}(X_{n_k})\to\infty$, a contradiction. This shows that $\sup_n\abs{\mathbb{E}[X_n]}<\infty$. Now, Proposition \ref{prop: preliminary result} asserts that for every $n\in\N$ there exists $m_n\in\R$ such that
\[
\pi_\alpha(X_n,m_n)=\pi_\alpha(X_n).
\]
We claim that $(m_n)$ is bounded. To this effect, suppose otherwise that $(m_n)$ is unbounded. Extract a subsequence $(m_{n_k})$ such that $\abs{m_{n_k}} \rightarrow \infty$. Since $\sup\abs{\mathbb{E}[X_n]}<\infty$, we must have $\abs{m_{n_k}-\mathbb{E}[X_{n_k}]}\to\infty$. Note that, by \eqref{DMeqaa}, we have
\[
\pi_\alpha(X_{n_k},m_{n_k})=N_{\alpha}((X_{n_k}-m_{n_k})^+)+m_{n_k} \geq N_{\alpha}((\mathbb{E}[X_{n_k}]-m_{n_k})^+)+m_{n_k}-\mathbb{E}[X_{n_k}]
+\mathbb{E}[X_{n_k}].
\]
Hence, we can write
\[
\pi_\alpha(X_{n_k},m_{n_k}) \geq \pi_{\alpha}(0,m_{n_k}-\mathbb{E}[X_{n_k}])+\mathbb{E}[X_{n_k}].
\]
Since $\abs{m_{n_k}-\mathbb{E}[X_{n_k}]}\to\infty$, it follows from Proposition~\ref{prop: preliminary result} that $\pi_{\alpha}(X_{n_k})\to\infty$, a contradiction. This proves that $(m_n)$ is bounded. To conclude the proof, we extract a subsequence $(m_{n_k})$ such that $m_{n_k} \rightarrow m \in \mathbb{R}$. Then, $(X_{n_k}-m_{n_k})^+ \xrightarrow{{a.s.}} (X-m)^+$. It follows from Zaanen~\cite[Theorem 131.6]{Zaanen1983} that
\[
N_{\alpha}((X-m)^+)\leq \liminf_k N_{\alpha}((X_{n_k}-m_{n_k})^+).
\]
As a consequence, we finally obtain
\begin{align*}
\pi_{\alpha}(X)
\leq&
\,\,N_{\alpha}((X-m)^+)+m \leq \liminf_k(N_{\alpha}((X_{n_k}-m_{n_k})^+)  +m_{n_k}) \\
=&
\,\liminf_k\pi_\alpha(X_{n_k},m_{n_k})=\liminf_k\pi_\alpha(X_{n_k})=\lim_n\pi_\alpha(X_{n})
=\liminf_n \pi_{\alpha}(X_n).
\end{align*}
This establishes {\em (iii)} and completes the proof.
\end{proof}

\smallskip

We proceed to study the stronger Lebesgue property and continuity from above. We write $X_n\downarrow X$ to mean that the sequence $(X_n)\subset L^\Phi$ is decreasing and converges almost surely to $X\in L^\Phi$.

\begin{theorem}
\label{prop: lebesgue continuity}
For the Haezendonck-Goovaerts premium principle $\pi_\alpha$ the following statements are equivalent:
\begin{enumerate}[(i)]
  \item $\pi_\alpha$ is continuous from above, i.e.\ for every sequence $(X_n)\subset L^\Phi$ and every $X\in L^\Phi$ we have
\[
X_n\downarrow X \ \implies \ \pi_\alpha(X_n)\to\pi_\alpha(X).
\]
  \item $\pi_\alpha$ satisfies the Lebesgue property, i.e.\ for every sequence $(X_n)\subset L^\Phi$ such that $\sup|X_n|\in L^\Phi$ and for every $X\in L^\Phi$ we have
\[
X_n\xrightarrow{a.s.}X \ \implies \ \pi_\alpha(X_n)\to\pi_\alpha(X).
\]
  \item $\Phi$ satisfies the $\Delta_2$ condition.
\end{enumerate}
\end{theorem}
\begin{proof}
Suppose that $\Phi$ satisfies the $\Delta_2$ condition and recall that $L^\Phi=H^\Phi$ in this case. That $\pi_\alpha$ satisfies the Lebesgue property follows from Bellini and Rosazza Gianin~\cite[Proposition 4]{BelliniRosazza2012} and Orihuela and Ruiz Gal\`{a}n~\cite[Theorem 1]{OrihuelaRuiz2012}. This shows that {\em (iii)} implies {\em (ii)}. It is immediate to see that {\em (ii)} implies {\em (i)}. We conclude by showing that {\em (i)} implies {\em (iii)}. To this effect, assume that $\pi_\alpha$ is continuous from above but $\Phi$ does not satisfy the $\Delta_2$ condition. In this case, $\Phi_\alpha$ also fails to satisfy the $\Delta_2$ condition. Then, it follows from Zaanen~\cite[Theorem 133.4]{Zaanen1983} that we find a sequence $(X_n)\subset L^\Phi$ and a scalar $\delta\in(0,\infty)$ such that $X_n\downarrow0$ and $N_\alpha(X_n)\geq\delta$ for every $n\in\N$. By Proposition \ref{prop: preliminary result}, for every $n\in\N$ we find $m_n\in\R$ such that $\pi_\alpha(X_n,m_n)=\pi_\alpha(X_n)$. Note that
\begin{equation}
\label{eq: lebesgue continuity}
\pi_\alpha(X_n,m_n) = \pi_\alpha(X_n) \to \pi_\alpha(0) = 0
\end{equation}
by continuity from above. Since $\pi_\alpha(X_n,0)=N_\alpha(X_n)\geq\delta$ for every $n\in\N$, we see that only finitely many $m_n$'s can be zero. Hence, we can always assume without loss of generality that $(m_n)$ consists either of strictly-positive or of strictly-negative numbers. We first focus on the strictly-positive case. Assume that $m_n>0$ for every $n\in\N$. We claim that $m_n\to0$. Indeed, otherwise, there exist $\varepsilon\in(0,\infty)$  and a subsequence $(m_{n_k})$ such that $m_{n_k}\geq\varepsilon$ for every $k\in\N$. Then
\[
\pi_\alpha(X_{n_k},m_{n_k}) \geq \pi_\alpha(0,m_{n_k}) = m_{n_k} \geq \varepsilon
\]
for every $k\in\N$, which contradicts \eqref{eq: lebesgue continuity}. This proves the claim. Now, for each $n\in\N$ the affine function $f_n:\R\to\R$ defined by
\[
f_n(m) = \frac{\pi_\alpha(X_n,m_n)-\pi_\alpha(X_n,0)}{m_n}m+\pi_\alpha(X_n,0)
\]
is easily seen to satisfy $f_n(0)=\pi_\alpha(X_n,0)$ and $f_n(m_n)=\pi_\alpha(X_n,m_n)$. Hence, it follows from the convexity of $\pi_\alpha(X_n,\cdot)$ that $\pi_\alpha(X_n,-1)\geq f_n(-1)$. This yields
\[
\pi_\alpha(X_1,-1) \geq \pi_\alpha(X_n,-1) \geq f_n(-1) \geq \frac{\delta-\pi_\alpha(X_n,m_n)}{m_n}+\delta
\]
for every $n\in\N$. However, this is clearly not possible because the right-hand side explodes as $n\to\infty$. As a result, our initial assumption that $\Phi$ does not satisfy the $\Delta_2$ condition is not tenable. Next, we focus on the strictly-negative case. Assume that $m_n<0$ for every $n\in\N$. We again show that $m_n\to0$. Otherwise, there exist $\varepsilon\in(0,\infty)$  and a subsequence $(m_{n_k})$ such that $m_{n_k}\leq-\varepsilon$ for every $k\in\N$ and thus
\begin{align*}
\pi_\alpha(X_{n_k},m_{n_k})
\geq&
\,\,\pi_\alpha(0,m_{n_k})
=
m_{n_k}+N_\alpha(-m_{n_k})
=
(-m_{n_k})\left(\frac{1}{\Phi^{-1}(1-\alpha)}-1\right) \\
\geq&
\left(\frac{1}{\Phi^{-1}(1-\alpha)}-1\right)\varepsilon >
0
\end{align*}
for every $k\in\N$, contradicting \eqref{eq: lebesgue continuity} as well. Here, we have used the fact that
\[
\Phi^{-1}(1-\alpha) := \sup\big\{x\in[0,\infty) \,; \ \Phi(x)\leq1-\alpha\big\} \in (0,1).
\]
From now on we can argue as in the strictly-positive case. In sum, we have established that {\em (i)} must imply {\em (iii)}, concluding the proof.
\end{proof}

\smallskip

Since the Haezendonck-Goovaerts principle $\pi_\alpha$ is law invariant, it can be regarded as a functional defined on the set of probability measures associated to random variables in $L^\Phi$. It is therefore natural to also study continuity properties in this setting. In the spirit of Kr\"{a}tschmer et al.~\cite{KratschmerSchiedZahle2014}, we focus on (semi)continuity with respect to the so-called $\Phi$-weak convergence. In line with the language of this note, we formulate our statements in terms of random variables instead of probability measures. To this effect, we write $X_n\xrightarrow{dist.}X$ to mean that the sequence $(X_n)\subset L^\Phi$ converges in distribution to $X\in L^\Phi$. Recall from \cite{KratschmerSchiedZahle2014} that $\Phi$-weak convergence can be properly formulated only on the Young class
\[
Y^\Phi := \left\{X\in L^0 \,; \ \E\left[\Phi\left(|X|\right)\right]<\infty\right\}.
\]
Clearly, $H^\Phi\subset Y^\Phi\subset L^\Phi$. For a sequence $(X_n)\subset Y^\Phi$ and $X\in Y^\Phi$ we say that $(X_n)$ $\Phi$-converges in distribution to $X$ whenever (see \cite[Lemma A.1]{KratschmerSchiedZahle2014})
\[
X_n\xrightarrow{\Phi\text{-}dist.}X \ :\iff \ X_n\xrightarrow{dist.}X \ \ \mbox{and} \ \  \E[\Phi(|X_n|)]\to\E[\Phi(|X|)].
\]

\smallskip

Our first result shows that Haezendonck-Goovaerts premium principles are always lower semicontinuous with respect to $\Phi$-convergence in distribution.

\begin{proposition}
The Haezendonck-Goovaerts premium principle $\pi_\alpha$, restricted to $Y^\Phi$, is lower semicontinuous with respect to $\Phi$-convergence in distribution, i.e.\ for every sequence $(X_n)\subset Y^\Phi$ and every $X\in Y^\Phi$ we have
\[
X_n\xrightarrow{\Phi\text{-}dist.}X \ \implies \ \pi_\alpha(X)\leq \liminf \pi_\alpha(X_n).
\]
\end{proposition}
\begin{proof}
Let $(X_n)\subset Y^\Phi$ and $X\in Y^\Phi$ be such that $X_n\xrightarrow{\Phi\text{-}dist.}X$.
Since our probability space is nonatomic, the classical Skorohod representation yields $(Y_n)\subset L^\Phi$ and $Y\in L^\Phi$ such that $X$ and $Y$ have the same distribution, $X_n$ and $Y_n$ have the same distribution for every $n\in\N$, and $Y_n\xrightarrow{a.s.}Y$. Clearly,
\begin{equation}\label{wlseq}
\mathbb{E}[\Phi(|Y|)]= \mathbb{E}[\Phi(|X|)]= \lim\mathbb{E}[\Phi(|X_n|)]=\lim\mathbb{E}[\Phi(|Y_n|)]<\infty.
\end{equation}
Observe that $\inf_n \mathbb{E}[Y_n] >-\infty$. Indeed, if $\inf_n \mathbb{E}[Y_n]=-\infty$, then we find a subsequence $(Y_{n_k})$ such that $\mathbb{E}[|Y_{n_k}|] \rightarrow \infty$. Therefore, by Jensen's inequality, $\mathbb{E}[\Phi(|Y_{n_k}|)]\geq \Phi[\mathbb{E}[|Y_{n_k}|] \rightarrow \infty$, contradicting \eqref{wlseq}. As a consequence, Theorem~\ref{prop: fatou}(i) yields that $\pi_\alpha(X)=\pi_\alpha(Y)\leq \liminf\pi_\alpha(Y_n)=\liminf \pi_\alpha(X_n)$.
\end{proof}

\smallskip

In Kr\"{a}tschmer et al.~\cite[Theorem 2.8]{KratschmerSchiedZahle2014} it was proved that {\em all} law-invariant convex risk measures are continuous with respect to $\Phi$-convergence in distribution on the Orlicz heart $H^\Phi$ if and only if $\Phi$ is $\Delta_2$. The following result shows that it is enough to check continuity for the Haezendonck-Goovaerts premium principle $\pi_\alpha$ to know whether $\Phi$ is $\Delta_2$ or not.

\begin{proposition}
For the Haezendonck-Goovaerts premium principle $\pi_\alpha$ the following statements are equivalent:
\begin{enumerate}[(i)]
\item[(i)] $\pi_\alpha$ is continuous on $Y^\Phi$ with respect to $\Phi$-convergence in distribution.
\item[(ii)] $\pi_\alpha$ is continuous on $H^\Phi$ with respect to $\Phi$-convergence in distribution.
\item[(iii)] $\Phi$ satisfies the $\Delta_2$ condition.
\end{enumerate}
\end{proposition}
\begin{proof}
It follows from Kr\"{a}tschmer et al.~\cite[Theorem 2.8]{KratschmerSchiedZahle2014} that {\em (iii)} implies {\em (i)} because $H^\Phi=Y^\Phi=L^\Phi$ when $\Phi$ is $\Delta_2$. That {\em (i)} implies {\em (ii)} is obvious since $H^\Phi\subset Y^\Phi$. Hence, it remains to show that {\em (ii)} implies {\em (iii)}. To this end, suppose that {\em (ii)} holds but {\em (iii)} fails. As in the proof of Theorem~\ref{prop: lebesgue continuity}, we take a sequence $(X_n)\subset L^\Phi$ and a scalar $\delta\in(0,\infty)$ such that $X_n\downarrow0$ and $N_\alpha(X_n)\geq\delta$ for every $n\in\N$. By rescaling, we may assume without loss of generality that $\norm{X_1}_\Phi\leq 1$, so that $\E[\Phi(X_1)]\leq 1$; see, e.g., Edgar and Sucheston~\cite[Proposition 2.1.10]{EdgarSucheston1992}. For each $n\in\N$ we clearly have $X_n\wedge k\uparrow X_n$, so that $N_\alpha(X_n\wedge k)\uparrow N_\alpha(X_n)$; see, e.g., Edgar and Sucheston~\cite[Theorem 2.1.11]{EdgarSucheston1992}. As a result, for every $n\in\N$ we can take $a_n\in(0,\infty)$ such that $N_\alpha(X_n\wedge a_n)\geq\frac{\delta}{2}$. Now, set $Y_n=X_n\wedge a_n$ for $n\in\N$. Then, $(Y_n)\subset L^\infty\subset H^\Phi$. Clearly, $Y_n\xrightarrow{a.s.}0$, so that $\Phi(Y_n)\xrightarrow{a.s.}0$. Since $0\leq \Phi(Y_n)\leq \Phi(X_n)\leq \Phi(X_1)\in L^1$, the Dominated Convergence Theorem implies that $\E[\Phi(Y_n)]\to0$. This yields $Y_n\xrightarrow{\Phi\text{-}dist.}0$ and it therefore follows from our assumption {\em (ii)} that $\pi_\alpha(Y_n)\to \pi_\alpha(0)=0$. However, $N_\alpha(Y_n)\geq \frac{\delta}{2}$ and $0\leq Y_n\leq X_1$ for every $n\in\N$. The same argument used in the proof of {\em (i)}$\implies${\em (iii)} in Theorem~\ref{prop: lebesgue continuity} applies to $(Y_n)$ in place of $(X_n)$ there and yields a contradiction. This completes the proof.
\end{proof}

\smallskip

\begin{remark}
(i) One may wonder whether the preceding (semi)continuity results also hold with respect to the weaker convergence in distribution (which corresponds to the so-called weak convergence at the level of probability measures). It is easy to see that $\pi_\alpha$ is never continuous with respect to such convergence. Indeed, otherwise, take any $n\in\N$ and let $X_n=n1_{A_n}$ where $\mathbb{P}(A_n)=\frac{1}{n}$. Since we clearly have $X_n\xrightarrow{dist.}0$, it would follow that $0=\pi_\alpha(0)=\lim \pi_\alpha(X_n)\geq \lim\E[X_n]=1$ by \eqref{DMeq}, which is absurd. However, it remains open to us whether $\pi_\alpha$ is lower semicontinuous with respect to convergence in distribution. Incidentally, we note that lower semicontinuity with respect to convergence in distribution is related to a Fatou-type property that was first introduced in Gao and Munari~\cite{GaoMunari2017}. Indeed, one can readily show that the following statements are equivalent:
\begin{enumerate}
\item[(a)] $\pi_\alpha$ is lower semicontinuous with respect to convergence in distribution, i.e.\ for every sequence $(X_n)\subset L^\Phi$ and every $X\in L^\Phi$ we have
\[
X_n\xrightarrow{dist.}X \ \implies \ \pi_\alpha(X)\leq \liminf \pi_\alpha(X_n).
\]
\item[(b)] $\pi_\alpha$ satisfies the super Fatou property, i.e.\
 for every sequence $(X_n)\subset L^\Phi$ and every $X\in L^\Phi$ we have
\[
X_n\xrightarrow{a.s.}X \ \implies \ \pi_\alpha(X)\leq \liminf \pi_\alpha(X_n).
\]
\end{enumerate}
It is immediate that {\em (a)} implies {\em (b)} because almost-sure convergence implies convergence in distribution. Conversely, assume that {\em (b)} holds. Let $(X_n)\subset L^\Phi$ and $X\in L^\Phi$ be such that $X_n\xrightarrow{dist.}X$. The classical Skorohod representation yields $(Y_n)\subset L^\Phi$ and $Y\in L^\Phi$ such that $X$ and $Y$ have the same distribution, $X_n$ and $Y_n$ have the same distribution for every $n\in\N$, and $Y_n\xrightarrow{a.s.}Y$. Thus, by {\em (b)}, we obtain
$\pi_\alpha(X)=\pi_\alpha(Y)\leq \liminf\pi_\alpha(Y_n)=\liminf \pi_\alpha(X_n)$ as desired.

\medskip

(ii) Recall that for every $X\in L^1$ the Expected Shortfall of $X$ at level $\lambda\in(0,1)$ is defined by
\[
\ES_\lambda(X) := \frac{1}{1-\lambda}\int_\lambda^1\VaR_p(X)dp,
\]
where $\VaR_p(X)$ denotes the $p$-lower quantile of $X$ for every $p\in(0,1)$. It is worth mentioning that, when $\Phi(x)=x$ on $[0,\infty)$, we have $L^\Phi=L^1$ and $\pi_\alpha$ coincides with Expected Shortfall via
\[
\pi_\alpha(X) = \inf_{m\in\R}\left\{m+\frac{1}{1-\alpha}\E[(X-m)^+]\right\} = \ES_\alpha(X)
\]
for every $X\in L^1$. In this case, it was shown in Example 2.14 of Chen et al.\ \cite{Chen2018} that $\pi_\alpha$ has the super Fatou property.
\end{remark}


\section{Dual representations}

In the previous section we have established that Haezendonck-Goovaerts premium principles satisfy the Fatou property on {\em any} Orlicz space, regardless of whether $\Phi$ fulfills the $\Delta_2$ condition or not. This allows us to exploit some recent results on dual representations of law-invariant functionals on Orlicz spaces to derive a dual characterization of Haezendonck-Goovaerts premium principles as worst expectations on a suitable set of probability measures. This extends Bellini and Rosazza Gianin~\cite[Proposition 4]{BelliniRosazza2012} beyond the $\Delta_2$ case. Recall that the conjugate function of $\pi_{\alpha}$ with respect to the duality pair $(L^\Phi,L^\Psi)$ is the map $\pi_{\alpha}^\ast:L^\Psi\to(-\infty,\infty]$ defined by
\[
\pi_{\alpha}^\ast(Y) := \sup_{X\in L^\Phi}\{\E[XY]-\pi_{\alpha}(X)\}.
\]

\begin{lemma}
\label{prop_conj}
For every $Y \in L^\Psi$ we have that
$$
\pi_{\alpha}^\ast(Y)=
\begin{cases}
0 & Y\geq 0, \ ||Y||_{\Psi_{\alpha}} \leq 1, \ \E[Y]=1, \\
+\infty & \mbox{otherwise}.
\end{cases}
$$
Here $\Psi_\alpha$ is the conjugate of $\Phi_\alpha$ and $||\cdot||_{\Psi_{\alpha}}$ is the Orlicz norm of $L^{\Psi_{\alpha}}$.
\end{lemma}
\begin{proof}
Take any $Y\in L^\Psi$ and define a map $\pi_{\alpha,H^\Phi}^\ast:L^\Psi\to(-\infty,\infty]$ by setting
\[
\pi_{\alpha,H^\Phi}^\ast(Y) = \sup_{X\in H^\Phi}\{\E[XY]-\pi_{\alpha}(X)\}.
\]
We clearly have that $\pi_{\alpha}^\ast(Y)\geq\pi_{\alpha,H^\Phi}^\ast(Y)$. We claim that the converse inequality also holds. To see this, let $X \in L^\Phi$. By Gao et al.~\cite[Proposition 3.4]{GaoLeungMunariXanthos2017} there exists a sequence $(\cF_n)$ of finitely-generated $\sigma$-subalgebras of $\cF$ such that $\E[X|\cF_n]\to X$ almost surely and $\sup_{n\in\N}|\E[X|\cF_n]|\in L^\Phi$. Since $\pi_{\alpha}$ has the Fatou property, we have
\[
\pi_{\alpha}(X) \leq \liminf\pi_{\alpha}(\mathbb{E}[X|\cF_n]).
\]
Now, note that $\cC=\{Y\in L^\Phi \,; \ \pi_{\alpha}(Y)\leq\pi_{\alpha}(X)\}$ is convex,
law-invariant, closed with respect to dominated almost-sure convergence by the Fatou property, and clearly contains $X$. Hence, by Gao et al.~\cite[Corollary 4.5]{GaoLeungMunariXanthos2017} we have $\mathbb{E}[X|\cF_n]\in\cC$ for every $n\in\N$, so that $\limsup\pi_{\alpha}(\mathbb{E}[X|\cF_n]) \leq \pi_{\alpha}(X)$. As a consequence, we infer that
\begin{equation}
\label{eq000}
\pi_{\alpha}(\mathbb{E}[X|\cF_n]) \to \pi_{\alpha}(X).
\end{equation}
Note that $\mathbb{E}[X|\cF_n]\in L^\infty\subset H^\Phi$ for each $n\in\N$. Thus, by applying~\eqref{eq000} and the Dominated Convergence Theorem, we get
\[
\pi_{\alpha,H^\Phi}^\ast(Y) \geq \lim\Big(\E[\E[X|\cF_n]Y]-\pi_{\alpha}(\E[X|\cF_n])\Big) = \E[XY]-\pi_{\alpha}(X),
\]
which in turn implies $\pi^\ast_{\alpha}(Y)\leq\pi_{\alpha,H^\Phi}^\ast(Y)$. Therefore, we have $\pi^\ast_{\alpha}(Y)=\pi_{\alpha,H^\Phi}^\ast(Y)$ as claimed. We can now rely on Bellini and Rosazza Gianin~\cite[Proposition 4]{BelliniRosazza2012} to infer that
\[
\pi_{\alpha}^\ast(Y)=
\begin{cases}
0 & Y\geq 0, \ ||Y||_{\Psi_{\alpha}} \leq 1, \ \E[Y]=1, \\
+\infty & \mbox{otherwise}.
\end{cases}
\]
This concludes the proof of the lemma.
\end{proof}

\smallskip

We can now state the announced dual representation of Haezendonck-Goovaerts premium principles. In what follows we denote by $\mathcal{P}(\probp)$ the set of all probability measures over $(\Omega,\mathcal{F})$ that are absolutely continuous with respect to $\probp$.

\begin{theorem}
The Haezendonck-Goovaerts premium principle $\pi_\alpha$ satisfies
\[
\pi_{\alpha}(X) = \sup_{\mathbb{Q}\in\mathcal{Q}^\infty_{\alpha}}\E_\mathbb{Q}[X]
\]
for every $X\in L^\Phi$, where
\[
\mathcal{Q}^\infty_{\alpha} = \bigg\{\mathbb{Q}\in\mathcal{P}(\probp) \,; \ \frac{d\mathbb{Q}}{d\mathbb{P}}\in L^\infty, \   \bigg\|\frac{d\mathbb{Q}}{d\mathbb{P}}\bigg\|_{L^{\Psi_{\alpha}}}\leq 1\bigg\}.
\]
\end{theorem}
\begin{proof}
Recall that $\pi_{\alpha}$ is convex, law-invariant, and satisfies the Fatou property on $L^\Phi$. Then, it follows from Gao at el.~\cite[Theorem 1.1]{GaoLeungMunariXanthos2017} that $\pi_{\alpha}$ is $\sigma(L^\Phi,L^\infty)$ lower semicontinuous. Thus, a direct application of the classical Fenchel-Moreau dual representation gives
\[
\pi_\alpha(X) = \sup_{Y\in L^\infty}\{\E[XY]-\pi_\alpha^\ast(Y)\}
\]
for every $X\in L^\Phi$. In view of Lemma~\ref{prop_conj}, we can write
\[
\pi_\alpha(X) = \sup\{\E[XY] \,; \ Y\in L^\infty, \ Y\geq 0, \ ||Y||_{\Psi_{\alpha}}\leq 1, \ \E[Y]=1\}
\]
for every $X\in L^\Phi$. It remains to observe that every element $Y$ in the above supremum can be expressed as a Radon-Nikodym derivative with respect to $\probp$.
\end{proof}

\smallskip

\begin{remark}
In the theorem, if we use $
\mathcal{Q}^H_{\alpha} = \big\{\mathbb{Q}\in\mathcal{P}(\probp) \,; \ \frac{d\mathbb{Q}}{d\mathbb{P}}\in H^\Psi, \   \big\|\frac{d\mathbb{Q}}{d\mathbb{P}}\big\|_{L^{\Psi_{\alpha}}}\leq 1\big\}$ in the supremum, then it follows from Theorem~\ref{prop: lebesgue continuity} and Orihuela and Ruiz Gal\`{a}n~\cite[Theorem 1]{OrihuelaRuiz2012} that the supremum is attained at every $X\in L^\Phi$ if and only if $\Phi$ satisfies the $\Delta_2$ condition.
\end{remark}


\section*{Acknowledgements}

We thank an anonymous referee for the insightful observations that led to include the discussion on $\Phi$-weak convergence. Niushan Gao and Foivos Xanthos acknowledge financial support of NSERC Discovery Grants.


\end{document}